\documentclass[fleqn]{eptcs}

\usepackage{pslatex,amssymb,amsmath,tikz}
\usetikzlibrary{arrows,snakes}
\usepackage[mathscr]{eucal}

\newtheorem{definition}{Definition}
\newtheorem{example}{Example}
\newtheorem{theorem}{Theorem}
\newtheorem{lemma}[theorem]{Lemma}
\newtheorem{corollary}[theorem]{Corollary}

\newcommand{\qed}{\phantom{a}\hfill $\Box$}
\newenvironment{proof}{\noindent\emph{Proof.}}{\qed}

\newcommand{\fif}{\rmiff}
\newcommand{\rmall}{\mathrel{\mbox{for all}~}}
\newcommand{\rmsome}{\mathrel{\mbox{for some}~}}
\newcommand{\rmwhere}{\mathbin{\mbox{where}~}}

\newcommand{\rmin}{\mathbin{~\mbox{in}~}}
\newcommand{\rmand}{\mathbin{~\mbox{and}~}}
\newcommand{\rmiff}{\mathbin{~\mbox{iff}~}}

\renewcommand{\mod}{\mathbin{\textsf{mod}}}

\newcommand{\Math}[1]{{${#1}$}}

\newcommand{\name}{\emph}
\newcommand{\class}{\textsc}
\newcommand{\auth}{\textsf}
\newcommand{\book}{\textit}

\newcommand{\jour}{\textsl}
\newcommand{\vol}{}
\newcommand{\nr}[1]{.#1}
\newcommand{\edt}[1]{(\auth{#1}, ed.)}
\newcommand{\eds}[1]{(\auth{#1}, eds.)}
\newcommand{\pub}[2]{ (#1, #2)}
\newcommand{\yr}[1]{, #1}
\newcommand{\pp}[2]{, #1--#2}

\newcommand{\true}{\mathit{true}}
\newcommand{\false}{\mathit{false}}

\newcommand{\limplies}{\supset}
\newcommand{\eqvt}{\Leftrightarrow}

\newcommand{\sumover}[1]{{\displaystyle \Sigma_{#1}}}
\newcommand{\sumbounds}[2]{{\displaystyle \Sigma_{#1}^{#2}}}
\newcommand{\orover}[1]{{\displaystyle \bigvee}_{#1}}

\newcommand{\alphbt}{\mbox{$A$}}
\newcommand{\oprt}{\sim}

\newcommand{\st}{.}

\newcommand{\intv}{\mbox{$\mathit{Intv}_w$}}
\newcommand{\pos}{\mbox{$\mathit{Pos}_w$}}
\newcommand{\lpos}{\mbox{$\ell\pos$}}
\newcommand{\cpos}{\mbox{$c\pos$}}

\newcommand{\btlth}{\mbox{$\mathit{BThTL}$}}
\newcommand{\utlinv}{\mbox{$\mathit{InvTL}$}}
\newcommand{\btlinv}{\mbox{$\mathit{BInvTL}$}}
\newcommand{\utlmodinv}{\mbox{$\mathit{InvModTL}$}}
\newcommand{\blintl}{\mbox{$\mathit{BLinTL}$}}
\newcommand{\utl}{\mbox{$\mathit{TL[\fut,\past]}$}}
\newcommand{\until}{\mathbin{\textsf{U}}}
\newcommand{\since}{\mathbin{\textsf{S}}}
\newcommand{\now}{\textsf{Now}~}
\newcommand{\fut}{\textsf{F}}
\newcommand{\henceforth}{\textsf{G}}
\newcommand{\past}{\textsf{P}}
\newcommand{\always}{\textsf{H}}
\newcommand{\nextt}{\textsf{X}}
\newcommand{\prev}{\textsf{Y}}

\newcommand{\ctl}{\mbox{$\mathit{CTL}$}}
\newcommand{\ltl}{\mbox{$\mathit{LTL}$}}

\newcommand{\ltlun}{\mbox{$\ltl[\fut,\past]$\/}}
\newcommand{\ltlunsuc}{\mbox{$\ltl[\fut,\past,\nextt,\prev]$\/}}

\newcommand{\first}[1]{\mbox{$F_{#1}$}}
\newcommand{\firstp}[1]{\mbox{$F^+_{#1}$}}

\newcommand{\last}[1]{\mbox{$L_{#1}$}}

\newcommand{\lastm}[1]{\mbox{$L^-_{#1}$}}
\newcommand{\firsta}{\first{a}}

\newcommand{\lasta}{\last{a}}

\newcommand{\uitl}{\mbox{$\mathit{UITL}$}}
\newcommand{\uitlpm}{\mbox{$\mathit{UITL^{\pm}}$}}

\newcommand{\fo}{\mbox{$\mathit{FO}$}}

\newcommand{\fotwo}{\mbox{$\fo^2$\/}}

\newcommand{\Pspace}{\class{Pspace}}

\newcommand{\Expspace}{\class{Expspace}}
\newcommand{\Exptime}{\class{Exptime}}
\newcommand{\twoexptime}{\class{2Exptime}}

\newcommand{\Np}{\class{Np}}

\newcommand{\oomit}[1]{}

\newcommand{\such}{.~}
\newcommand{\pti}{\aloo{~}}

\newcommand{\aloo}[1]{\lceil #1 \rceil}

\newcommand{\succr}{\oplus}
\newcommand{\predr}{\ominus}

\newcommand{\potdfa}{\mbox{$\mathit{po2dfa}$}}

\newcommand{\rpotdfa}{\mbox{$\mathit{RecPO2DFA}$}}

\newcommand{\rectl}{\mbox{$\mathit{TL}[X_\phi,Y_\phi]$}}
\newcommand{\tlrecr}{\mbox{$\mathit{TL}^+[X_\phi,Y_\phi]$}}

\newcommand{\at}{\mbox{$\mathit{At}$}}

\newcommand{\tlxy}{\mbox{$\mathit{TL[X_a,Y_a]}$}}
\newcommand{\weakx}[1]{\widetilde{X}_{#1}} 
\newcommand{\weaky}[1]{\widetilde{Y}_{#1}}
\newcommand{\xunit}{\nextt}
\newcommand{\yunit}{\prev}
\newcommand{\pfless}{\mathcal P^{<}}
\newcommand{\pfleq}{\mathcal P^{\leq}}
\newcommand{\pfgreat}{\mathcal P^{>}}
\newcommand{\pfgeq}{\mathcal P^{\geq}}
\newcommand{\gabar}[1]{G_{\overline{#1}}}
\newcommand{\habar}[1]{H_{\overline{#1}}}

\newcommand{\succp}{\overline{\oplus}}
\newcommand{\predm}{\overline{\ominus}}

\newcommand{\Nat}{{\mathbb{N}}}
\newcommand{\Int}{{\mathbb{Z}}}

\title{Deterministic Temporal Logics and Interval Constraints}
\author{Kamal Lodaya
\institute{The Institute of Mathematical Sciences, CIT Campus, 
Chennai \textit{600113}\thanks{The author is affiliated to 
Homi Bhabha National Institute, Anushaktinagar, Mumbai \textit{400094}.}}
\and Paritosh K.~Pandya
\institute{Tata Institute of Fundamental Research, Colaba, 
Mumbai \textit{400005}}
}
\def\titlerunning{Deterministic Temporal Logics and Interval Constraints}
\def\authorrunning{Kamal Lodaya and Paritosh K.~Pandya}

\begin{document}
\maketitle

Temporal logics have gained prominence in computer science 
as a property specification language for reactive systems. 
There is even an IEEE standard temporal logic supported by a consortium of 
Electronic Design Tool developers. 
Such systems maintain ongoing interaction between the environment and 
the system and their specification requires formulating constraints on 
the sequence of steps performed by the system. 
Unlike Classical logics which explicitly use variables to range over 
time points, temporal logics, which are rooted in tense logics,  
provide a variable-free approach which deals with time implicitly, using
modalities. 
The work on temporal logic for specifying and proving concurrent programs 
began with Pnueli's initial identification of this logic for reactive systems 
\cite{Pnu77}. Lamport also used
temporal logic to reason about properties of distributed systems \cite{Lam80}.

We work in the setting of finite and infinite words over a finite alphabet.  
A diverse set of modalities can be formulated to give different temporal logics.
However, over time, the linear temporal logic $\ltl$ has emerged as 
a standard formulation. 
A major driver for this choice is its economy of operators while 
being expressive;
it just uses modalities $\until$ and $\since$. 
The classical result of Kamp showed that the $\ltl$
logic is expressively complete with respect to $\fo$-definable
properties of words \cite{Kamp}. Moreover, as shown by Sistla and Clarke, 
the logic has elementary \Pspace-complete satisfiability \cite{SC}. 
Yet another class of temporal logics which provides very natural form of 
specification are the interval temporal logics. 
However, their high satisfaction complexity has prevented their widespread use.

It can be seen from these developments that the concerns for 
expressive power of the temporal logic and its algorithmic complexity 
have been major drivers. They directly affect the usability of model checking
tools developed. Several fragments/variants of $\ltl$ have been explored 
to improve its usability. For example, the industry standard PSL/Sugar 
adds regular expressions to $\ltl$. Various forms of counting constructs 
allowing quantitative constraints to be enforced have also been added to 
$\ltl$ and to interval temporal logics. At the same time, keeping 
algorithmic complexity in mind, fragments of $\ltl$ such as $\utl$ with 
low satisfaction complexity have been explored \cite{EVW02,IW09}.
But there are other possibilities.

One less-known such theme is that of ``deterministic logics''. 
In our own experience, while implementing a validity checker 
for an interval temporal logic over word models, we found marked improvements 
in  efficiency when nondeterministic modal operators were replaced by  
deterministic or unambiguous ones \cite{KP}. 
This led to our interest in results on unambiguous languages, 
initiated by Sch\"utzenberger \cite{Sch76}. In a subsequent paper \cite{LPS08} 
we learnt that these could also be thought of as boolean combinations 
of deterministic and co-deterministic products over a small class, 
the piecewise testable languages.
We expanded the scope of our work to studying determinism and guarding 
in modalities at all levels of temporal, timed and first-order logics.
This tutorial is a presentation of temporal logics with deterministic 
as well as guarded modalities, their expressiveness and computational efficiency. 

In Section~\ref{sec:dl}, we begin with deterministic modalities at 
the lowest level, using a couple of 
representative logics, and then we recursively build higher guarded 
deterministic modalities all the way to full temporal logic. 
These ideas were initiated by Kr\"oger \cite{Kro}. 
We cannot  claim that the deterministic modalities are the ones which will be 
preferred at the level of specification.  
Linear temporal logic \ltl\/ continues to be widely used.
What do its ``nondeterministic'' modalities buy for the user? 
In Section~\ref{sec:ic} we show that the introduction of guarded constraints 
specifying counting and simple algebraic operations over an interval, 
arguably an important part of specifying properties, at low modal depth, 
already reaches high levels of full temporal logic, while retaining 
elementary decidability. 

\paragraph{Temporal logics}
When talking of languages (over finite and infinite words), modal logics 
specialize to temporal logics. Words are nothing but rooted coloured 
linear orders where positions in the words denote possible worlds. 
Classically, modalities  $\fut,\past,\nextt,\prev,\until,\since$ 
are widely used. Their semantics is given below. 

Let $w \in \alphbt^+ \cup \alphbt^\omega$ be a word (finite or infinite). 
Let $dom(w)$ denote the set of positions in the word, e.g.
$dom(aba) = \{1,2,3\}$, and for an infinite word $dom(w)= \Nat$. 
We define the semantics of linear temporal logic operators below.

\Math{w,i \models a} iff \Math{a \in w[i]}

\Math{w,i \models \nextt\phi} iff  
\Math{i+1 \in dom(w) \rmand w,i+1 \models \phi}

\Math{w,i \models \prev\phi} iff 
\Math{i-1 \in dom(w) \rmand w,i-1 \models \phi}

\Math{w,i \models \fut\phi} iff for some
\Math{m > i: w,m \models \phi}

\Math{w,i \models \past\phi} iff for some
\Math{m \leq i: w,m \models \phi}

\begin{tabular}{l}
\Math{w,i \models \phi~\until~\psi} iff 
for some \Math{m > i: w,m \models \psi} 
and for all \Math{i < l < m: w,l \models \phi}\\
\Math{w,i \models \phi~\since~\psi} iff 
for some \Math{m < i: w,m \models \psi} 
and for all \Math{m < l < i: w,l \models \phi}\\
\end{tabular} \\
We also have defined operators
$\henceforth\phi = \lnot\fut\lnot\phi$ and
$\always\phi = \lnot\past\lnot\phi$. 
We remark that the operators $\until$ and $\since$ as well as 
the derived $\fut$ and $\past$ operators used in this tutorial are all 
``strict''.  

Let $OPS$ be a set of temporal operators. Then, $TL[OPS]$ defines temporal logic formulae using only the operators from $OPS$ and
the boolean connectives. An interesting question is about the expressive power of such a logic $TL[OPS]$ for various choices of $OPS$. 
For example, operators $\fut,\past,\nextt,\prev$ can be defined using 
$\until, \since$. 
Hence $TL[\fut,\past,\nextt,\prev,\until,\since] \equiv TL[\until,\since]$.

\section{Deterministic Logics and Unambiguous Star-free Langauges}
\label{sec:dl}

Sch\"utzenberger first studied Unambiguous star-free regular languages ($UL$) \cite{Sch76} and gave an algebraic characterization for $UL$. 
Since then, several diverse and unexpected characterizations have emerged for this language class: 
$\Delta_2[<]$ in the quantifier-alternation hierarchy of first-order definable languages \cite{PW97}, 
the two-variable fragment $FO^2[<]$ \cite{TW-until} (without any restriction on quantifier alternation), 
and Unary Temporal Logic \utl\/ \cite{EVW02} are some of the logical characterizations that are well known. 
Investigating the automata for $UL$, Schwentick, Th\'erien and Vollmer
\cite{STV01}  defined Partially Ordered 2-Way Deterministic Automata ($\potdfa$) 
and showed that these exactly recognize the language class $UL$. 
Weis and Immerman  have characterized $UL$ as a boolean combination of ``rankers'' \cite{IW09}.
A survey paper \cite{DGK08} describes this language class and its characterizations.

We go back to Sch\"utzenberger's definition.
A monomial over an alphabet $\alphbt$ is a regular expression of the form $A_0^* a_1 \cdots a_nA_n^*$, 
where $A_i\subseteq \alphbt$ and $a_i\in \alphbt$. 
By definition, $UL$ is the subclass of star-free regular languages which may be expressed as a finite disjoint union of unambiguous monomials: every word that belongs to the language, may be \emph{unambiguously} parsed so as to match a monomial. The uniqueness with which these monomials parse any word is the characteristic property of this language class. 
We explore a similar phenomenon in logics  by introducing the notion of \emph{Deterministic Temporal Logics}.

Given a modality $\mathscr M$ of a temporal logic that is interpreted over a word model, the \emph{accessibility relation} of 
$\mathscr M$ is a relation which maps every position in the word to the set of positions that are accessible by $\mathscr M$. In case of interval temporal logics, the relation is over intervals instead of positions in the word model. The modality is \emph{deterministic} if its accessibility relation is a (partial) function. A logic is said to be deterministic if all its modalities are deterministic. Hence, deterministic logics over words have 
the property of \emph{Unique Parsability} stated below.
\begin{definition}[Unique Parsability]
In the evaluation of a temporal logic formula over a given word, every subformula has a unique position (or interval) in the word at which it must be evaluated. This position is determined by the context of the subformula. 
\end{definition}

In this section, we investigate deterministic temporal logics and their properties. We give constructive reductions between deterministic logics with diverse modalities. We also analyze efficient algorithms for checking their satisfiability. For simplicity we confine ourselves to languages of finite words; the situation for languages of infinite words is not very different. We begin the study with a basic logic of rankers $\tlxy$, and investigate its satisfiability which turns out to be $NP$-complete. It is well known that $\tlxy$ exactly has
the expressive power of $UL$ \cite{IW09,STV01}.
We then look at deterministic interval logic $\uitlpm$ and give a polynomial-time reduction to $\tlxy$. This
reduction relies on a crucial property of ranker directionality investigated by Weis and Immerman \cite{IW09}
and others  \cite{PS14,SShah12}. Several  logics lie between these two logics and  they all have the same expressive power and 
$NP$-complete satisfaction complexity.

In order to go beyond $UL$, we consider a recursive extension of $\tlxy$. This deterministic logic was 
proposed by Kr\"oger and it has  been called AtNext logic in literature \cite{Kro}. We briefly investigate the 
relationship between $\ltl$ and the AtNext logic and show that both have the same expressive power.
However, the hierarchies induced by the two logics are quite different.

\subsection{The logic of Rankers}

We define the logic of rankers as follows.

\paragraph{Syntax}
\[
\phi ~:=~ 
\begin{array}[t]{l}
a ~\mid~ \top ~\mid~ { X_a \phi_1}  ~\mid~ { Y_a \phi_1}
~\mid~ { SP \phi_1}  ~\mid~  { EP\phi_1}
~\mid~\phi_1\lor\phi_2 ~\mid~ \neg\phi_1  ~

\mid~  \weakx{a} \phi_1 ~\mid~ \weaky{a} \phi_1
~\mid~\xunit \phi_1 ~\mid~ \yunit \phi_1 
\end{array}
\]
Let
$EP\phi = \lnot\nextt\top \limplies \phi$ be a derived operator.
For convenience, we have defined $\tlxy$ with many modalities. It can be shown (see \cite{SShah12})
that it is sufficient to have only $X_a$, $Y_a$ and $EP$ modalities; all other operators can be eliminated
giving an equivalent formula.

Let $Size(\phi)$ denote size (i.e. number of operators and atomic formulae) occurring in $\phi$.

\paragraph{Semantics} Given word $w \in A^+$ and $i \in dom(w)$ we have \\
\begin{tabular}{rcl}
$w,i\models \top$&& \\
$w,i\models X_a\phi$ & iff & $\exists j>i ~.~ w[j]=a$ and $\forall i<k<j. w[k]\neq a$ 
  and $w,j\models\phi$. \\
$w,i\models Y_a\phi$ & iff & $\exists j<i ~.~ w[j]=a$ and $\forall j<k<i. w(k[k]eq a$ 
  and $w,j\models\phi$.\\
$w,i\models \weakx{a}\phi$ & iff & $\exists j\geq i ~.~ w[j]=a$ and $\forall i\leq k<j. w[k]\neq a$ 
  and $w,j\models\phi$.\\
$w,i\models \weaky{a}\phi$ & iff & $\exists j\leq i ~.~ w[j]=a$ and $\forall j<k\leq i. w[k]\neq a$
  and $w,j\models\phi$.
\end{tabular}
\medskip

\noindent The language accepted by a \tlxy\/ formula $\phi$ is given by $\mathcal L(\phi)=\{w ~\mid~ w,1\models\phi\}$.
\begin{example}
Consider the $LTL$ formula $G(a \Rightarrow F b)$. This is equivalent to $\tlxy$ formula 
$\neg EP(Y_a \neg X_b \top)$.
\end{example}
\begin{example} Consider the Unambiguous monomial 
$\{a,c,d\}^* \cdot { c} \cdot \{a\}^* \cdot { b} \cdot \{a,b,c,d\}^*$. Then, its language
is equivalent to the language of $\tlxy$ formula $X_bY_c\neg( X_dX_b\neg Y_b T)$.
\end{example}

\begin{definition}[Ranker \cite{IW09}]
A ranker is a $\tlxy$ formula which does not use boolean operators $\neg, \land, \lor$ and it only has atomic formula
$\top$ (i.e. the use of atomic proposition $a$ is not allowed). 
\end{definition}
For example, $EP(Y_a \xunit \top)$ is a ranker. 

A ranker $RK$ (also called a turtle program \cite{STV01}) performs scans over a word $w$ which end at a position in the word or the scan fails.
The outcome of scan (i.e. last position) is denoted by $\lpos(RK) \in dom(w) \cup \{\bot\}$ where $\bot$ denotes the failure of the scan.  Note that the ranker search always starts at the initial position in the word. Thus, $\lpos(RK) = Pos(w,1,RK)$ where
\[
\begin{array}{l}
 Pos(w,i,\top) = i \\
 Pos(w,i,SP(RK)) = Pos(w,1,RK) \\ 
 Pos(w,i,\weakx{a}(RK)) = Pos(w,j,RK) ~~\mbox{if}~~j \geq i \land w[j]=a \rmand \forall i \leq k<j: w[k]\neq a \\
 Pos(w,i,\weakx{a}(RK)) = \bot ~~\mbox{if}~~ \forall j \st j \geq i \Rightarrow w[j] \not= a \\
 Pos(w,i,X_a(RK)) = Pos(w,j,RK) ~~\mbox{if}~~j>i \land w[j]=a \rmand \forall i<k<j: w[k]\neq a \\
 Pos(w,i,X_a(RK)) = \bot ~~\mbox{if}~~ \forall j \st j > i \Rightarrow w[j] \not= a \\
 Pos(w,i,\xunit(RK)) = Pos(w,i+1,RK) ~~~\mbox{if}~~i+1 \in dom(w) \\
 Pos(w,i,\xunit(RK)) = \bot ~~~\mbox{if}~~i+1 \notin dom(w) \\
\end{array}
\]
The remaining cases are similar and omitted.

Consider a formula $\phi$ and its subformula  $\beta$ occurring in context $\alpha[-]$, i.e. $\phi = \alpha[\beta]$. We shall call such $\alpha[\beta]$ as a subterm. 
With each subterm, we associate a ranker  denoted $Ranker(\alpha[ ~ ])$ which identifies the unique position
in word where subformula $\beta$ needs to be evaluated. This ranker does not depend on the subformula $\beta$ but only on the context $\alpha[ ~]$.
We give rules for calculating the Ranker of a subterm.
\[
 \begin{array}{l}
   Ranker([~ ]) = SP \top \\
   Ranker(\alpha(OP [ ~])) = RK( OP \top) ~~\mbox{where}~~ Ranker(\alpha[ ~]) = RK \top ~~\mbox{and} \\
          \hspace*{5cm}  OP \in \{ X_a, Y_a, \weakx{a}, \weaky{a}, \xunit, \yunit, SP, EP \} \\
   Ranker(\alpha( \beta_1 \lor [ ~])) =   Ranker(\alpha( \beta_1 \land  [ ~])) =   Ranker(\alpha[ ~])  \\
   Ranker(\alpha( \neg [ ~])) =   Ranker(\alpha[ ~])  \\
 \end{array}
\]

The main lemma below relates truth of atomic formulae at their ranker positions to the truth of the whole formula. 
\begin{lemma}[unique parsing]
\label{lem:xymodelchecking}
 Let $\phi$ be a formula of $\tlxy$ and let $t_i = \alpha_i[\beta_i]$ for $1 \leq i \leq k$ be all its subterms such that each 
 $\beta_i$ is an atomic formula (of the form $a$ or $\top$). Consider the witness propositional formula $W$ obtained by replacing each such subformula by 
 propositional letter $p_i$, and by omitting all the temporal operators but keeping all the boolean operators. Also, for any $w \in A^+$
 let $\mu_w$ be a valuation assiging $p_i=true$ iff $\lpos(Ranker(\alpha_i[~])=j\not=\bot \land w[j] \models_{prop} \beta_i$. Thus, valuation 
 $\mu$ records whether
 atomic formula $\beta_i$ holds at its ranker position. Then, $w,1 \models \phi$ iff $\mu \models_{prop} W$. 
\end{lemma}

\begin{example}
 Consider formula $ \phi = EP(Y_a ( \neg X_b \top ~~\lor~~\xunit c))$. Then, we have atomic subformulae 
 (occurrences) $\beta_1=\top$ and $\beta_2=c$ 
 with corresponding rankers $RK_1 = EP(Y_a X_b \top)$ and $RK_2 = EP(Y_a \xunit \top)$. 
The witness propositonal formula is $W$ is $(\neg p_1 \lor p_2)$. 
Consider a word $w = a b a d b c$. Then, $\lpos(RK_1)=5$ and $\lpos(RK_2)=4$. 
Hence $\mu_w(p_1)=\true$ and $\mu_w(p_2)=\false$.
 It is easy to see that $\mu_w \not \models W$. It is also clear that $w,1 \not \models \phi$.
\end{example}

\begin{corollary}
\label{coro:xytruthchecking}
Checking whether $w,1 \models \phi$ can be carried out in time $|w| \times |\phi|^3$.
\end{corollary}
\begin{proof}
Given the word, checking whether an atomic formula is true or false at its ranker position can be done in time $|w| \times |\phi|$. 
Number of such atomic formulae are linear in the size of $|\phi|$. This determines $\mu_w$. Given $\mu_w$, 
evaluating the propositional formula $W$ which is atmost of  size $|\phi|$ will take time at most linear in size of $\phi$. 
\end{proof}

We now establish a small model property for logic $\tlxy$.
\begin{lemma}
\label{lem:xysmallmodel}
 Let $\phi$ be a formula of $\tlxy$. 
 If $\phi$ is satisfiable then there exists $w$ with length $|w| = Size(\phi)$ such that $w,1 \models \phi$.
\end{lemma}
\begin{proof}
 Let $Rankerset(\phi)$ denote the set of rankers associated with each subterm of $\phi$. It is clear that
 size of $Rankerset(\phi)$ is at most $Size(\phi)$. We now define all the positions which are characterized by rankers. Since
 ranker scan starts at position $1$, this is always included in our set.
 Consider 
 $Rankersetpos_w(\phi) = \{ \lpos(RK) ~\mid~ RK \in Rankerset(\phi)\} \cup \{1\} - \{ \bot\} $. Let 
 $v = w \downarrow Rankersetpos_w(\phi)$ denote the word obtained by removing letters not at positions in $Rankersetpos_w(\phi)$. Hence size of $v$ is at most $Size(\phi)$.
 Also let $f: dom(w) \rightarrow dom(v)$ give the mapping of an undeleted position in $w$ to its corresponding
 position in $v$. Then, it is easy to see that $f(\lpos(RK)) = \mathit{\ell Pos}_v(RK)$ for each $RK \in Rankersetpos_w(\phi)$. This can 
 formally be proved by induction on the length of the ranker.
 From this and Lemma~\ref{lem:xymodelchecking} it is clear that $w,1 \models \phi ~~\fif~~ v,1 \models \phi$. Thus, $\phi$ has a linear sized model if it has a model.
\end{proof}

\begin{theorem}
 Satisfiability of $\tlxy$ is NP-complete.
\end{theorem}
\begin{proof}
 By Lemma \ref{lem:xysmallmodel}, we can nondeterministically guess a small word of size linear in size of $\phi$.
 Note that number of bits needed to represent this is $|\phi| log |\phi|$ since alphabet cannot be larger than
 the size of $\phi$.
 Checking that $w,1 \models \phi$ can be done in time polynomial in $w$ and $\phi$ by 
Corollary~\ref{coro:xytruthchecking}. Thus, satisfiability is in $NP$. 
Since logic $\tlxy$ includes propositional formulae, its satisfiability is also $NP$-hard.
\end{proof}

\subsection{Deterministic Interval Logic \uitlpm}
Now we consider a seemingly much more powerful deterministic interval temporal logic $\uitlpm$
(based on a logic in \cite{LPS}). 
We show that this logic can be reduced to
$\tlxy$ in polynomial time preserving models. This reduction also makes use of rankers and an additional critical property
called ranker directionality.

\label{sec:uitlpm}
In this section, we introduce the logic \uitlpm\/  and show that it is no more expressive than $UL$, by giving an effective conversion from \uitlpm\/ formulas to their corresponding language-equivalent \tlxy\/ formula. The conversion is similar to the conversion from \uitl\/ to \tlxy\/, as given in \cite{DKL10}.

\subsubsection{\uitlpm: Syntax and Semantics}
The syntax and semantics of \uitlpm\/ are as follows:
\[
\begin{array}{l}
D ::= \top ~\mid~ a ~\mid~ \pti ~\mid~ \mathit{unit} ~\mid~ SP\phi ~\mid~ EP\phi ~\mid~
D_1 \firsta D_2 ~\mid~ D_1 \lasta D_2 ~\mid~ D_1 \firstp{a} D_2 ~\mid~D_1 \lastm{a} D_2 ~\mid ~\\
\hspace{1cm} \succr D_1 ~\mid~ \predr D_1 ~\mid~ 
\succp D_1 ~\mid~ \predm D_1 ~\mid~ 
D_1 \lor D_2 ~\mid~ \neg D 
\end{array}
\]

Let $w$ be a nonempty finite word over $\alphbt$ 
and let $dom(w) = \{ 1, \ldots, |w|\}$ 
be the set of positions. Let 
$INTV(w) = \{ [i,j] ~\mid~ i,j \in dom(w), i \leq j \}~\cup~ \{\bot\}$ be the set of 
intervals over $w$, where $\bot$ is a special symbol to denote an undefined interval. For an interval $I$, let $l(I)$ and $r(I)$ denote the left and right endpoints of $I$. Further, if $I=\bot$, then $l(I)=r(I)=\bot$. The satisfaction of a formula $D$ is defined
over intervals of a word model $w$ as follows. \\
\[\begin{array}{l}
w,[i,j] \models \top \rmiff [i,j]\in INTV(w) \rmand [i,j]\neq\bot\\
w,[i,j]\models \pti \rmiff i=j\\
w,[i,j]\models \mathit{unit} \rmiff j=i+1\\
w,[i,j]\models SP\phi \rmiff w,[i,i]\models\phi\\
w,[i,j]\models EP\phi \rmiff w,[j,j]\models\phi
\end{array}
\]
\[\begin{array}{l}
w,[i,j] \models D_1 \firsta D_2 \rmiff \rmsome k: i \leq k \leq j \st ~~  
w[k]=a \rmand \\
\hspace*{1cm}  (\rmall m: i \leq m < k \st w[m] \neq a) \rmand 
    w,[i,k] \models D_1  \rmand w,[k,j] \models D_2  \\
w,[i,j] \models D_1 \lasta D_2 \rmiff  \rmsome k: i \leq k \leq j \st ~~  
w[k]=a \rmand  \\ 
\hspace*{1cm} (\rmall m: k < m \leq j \st w[m] \neq a) \rmand  
    w,[i,k] \models D_1  \rmand w,[k,j] \models D_2  \\
w,[i,j] \models D_1 \firstp{a} D_2 \rmiff \rmsome k: k \geq j \st ~~  
w[k]=a \rmand \\
\hspace*{1cm}  (\rmall m: i \leq m < k \st w[m] \neq a) \rmand 
w,[i,k] \models D_1  \rmand w,[j,k] \models D_2  \\
w,[i,j] \models D_1 \lastm{a} D_2 \rmiff  \rmsome k: k \leq i \st ~~  
w[k]=a \rmand  \\ 
\hspace*{1cm} (\rmall m: k < m \leq j \st w[m] \neq a) \rmand
w,[k,i] \models D_1  \rmand w,[k,j] \models D_2  \\
w,[i,j] \models \succr D_1 \rmiff i< j \rmand w,[i+1,j]\models D_1\\
w,[i,j] \models \predr D_1 \rmiff i< j \rmand w,[i,j-1]\models D_1\\
w,[i,j] \models \succp D_1 \rmiff j<|w|\rmand w,[i,j+1]\models D_1\\
w,[i,j] \models \predm D_1 \rmiff i>1 \rmand w,[i-1,j]\models D_1\\
\end{array}
\]
The language $\mathcal L(\phi)$ of a \uitlpm\/ formula $\phi$  is given by 
$\mathcal L(\phi) = \{w ~ \mid ~ w,[1,|w|] \models \phi\}$. 
Define $\aloo{A} = \pti~\lor~\mathit{unit}~\lor~\neg\bigvee\limits_{b\not\in A}(\succr\predr(\top F_b \top))$.
Hence, $w,[i,j]\models \aloo{A}$ if and only if $\forall i<k<j ~.~ w[k]\in A$.

\begin{example}
The language of unambiguous monomial 
$\{a,c,d\}^* \cdot { c} \cdot \{a\}^* \cdot { b} \cdot \{a,b,c,d\}^*$.
given earlier may be specified by the \uitlpm\/ formula $(\top ~L_c~ \aloo{a} )~F_b~\top)$. On the
other hand, the
formula $\phi ~=~ (\top L_b { (\neg \aloo{\neg c})} ) L_a \top$ states
that between last $a$ and its previous $b$ there is at least one $c$.
\end{example}

\subsubsection{Ranker Directionality}

Given a ranker and a word, it is possible to define by a $\tlxy$ formula whether we are to the left or right of the ranker's
characteristic position. This is called \emph{ranker directionality}. 
This property of rankers was investigated by Weis and Immerman \cite{IW09}
and Dartois, Kufleitner, Lauser \cite{DKL10}.

For a ranker formula { $\psi$}, we can define $\tlxy$ formulae 
{ $\pfless(\psi)$, $\pfleq(\psi)$, $\pfgreat(\psi)$, $\pfgeq(\psi)$}
satisfying the following lemma.
\begin{lemma}[Ranker Directionality] 
$\forall w\in\alphbt^+$ and $\forall i\in dom(w)$, if $\lpos(\psi)\neq\bot$, then
\begin{itemize}
\item $w,i\models { \pfless}(\psi)$ iff $i<\lpos(\psi)$
\item $w,i\models{ \pfleq}(\psi)$ iff $i\leq \lpos(\psi)$
\item $w,i\models { \pfgreat}(\psi)$ iff $i>\lpos(\psi)$
\item $w,i\models { \pfgeq}(\psi)$ iff $i\geq \lpos(\psi)$
\end{itemize}
{Moreover, the  Size of these formulae are linear in size of $\psi$.}
\end{lemma}

We give the construction of these formulae below and we omit the proof of above lemma which can be found in
\cite{DKL10,PS14,SShah12}.
\\
\begin{tabular}{|c|c|c|c|c|}
\hline
$\psi$ & $\pfless(\psi)$ & $\pfleq(\psi)$ & $\pfgreat(\psi)$ & $\pfgeq(\psi)$\\
\hline
\hline
$\phi SP\top$& $\bot$ & $\mathit{Atfirst}$ & $\neg\mathit{Atfirst}$ & $\top$\\
\hline
$\phi EP\top$& $\neg\mathit{Atlast}$ & $\top$ & $\bot$ & $\mathit{Atlast}$\\
\hline
$\phi\weakx{a}\top$ & $X_a(\pfleq(\psi))$ & $\habar{a}\lor (Y_a\pfless(\phi\top))$ & $Y_a\pfgeq(\phi\top)$ & $\gabar{a}\lor X_a\pfgreat(\psi)$\\
\hline

$\phi X_{a}\top$ & $X_a(\pfleq(\psi))$ & $\habar{a}\lor (Y_a\pfleq(\phi\top))$ & $Y_a\pfgreat(\phi\top)$ & $\gabar{a}\lor X_a\pfgreat(\psi)$\\
\hline
$\phi\weaky{a}\top$ & $X_a\pfleq(\phi\top)$ & $\habar{a}\lor (Y_a\pfless(\psi))$ & $Y_a\pfgeq(\psi)$ & $\gabar{a}\lor X_a\pfgreat(\phi\top)$\\
\hline
$\phi Y_a\top$ & $X_a\pfless(\phi\top)$ & $\habar{a}\lor (Y_a\pfless(\psi))$ & $Y_a\pfgeq(\psi)$ & $\gabar{a}\lor X_a\pfgeq(\phi\top)$\\
\hline
$\phi\xunit\top$ & $\pfleq(\phi\top)$ & $\mathit{Atfirst} ~\lor~ \yunit\pfleq(\phi\top)$ & $\yunit \pfgreat(\phi\top)$ & $\pfgreat(\phi\top)$ \\
\hline
$\phi\yunit\top$ & $\xunit\pfless(\phi\top)$ & $\pfless(\phi\top)$ & $\pfgeq(\phi\top)$ & $\mathit{Atlast} ~\lor~$ \\
  & & & & $\xunit\pfgeq(\phi\top)$ \\
\hline 
\end{tabular}

\subsubsection{Reducing $\uitlpm$ to $\tlxy$} Logic
\uitlpm\/ is a deterministic logic and the Unique Parsing property holds for its subformulas. 
Hence, for every \uitlpm\/ subformula $\psi$, and any word $w$, there is 
a unique interval $\intv(\psi)$ within which $\psi$ needs to be evaluated. 
Further, if subformula $\psi$ has as its major connective  a ``chop'' operator ($\firsta,\lasta,\firstp{a},\lastm{a},\succr,\predr,\succp,\predm$), then there is a unique chop position $\cpos(\psi)$. If such an interval or chop position fails to exist in the word, then we return $\bot$. The $\intv(\psi)$ and $\cpos(\psi)$ for any subformula $\psi$ depend on its context and rankers characterizing these positions  can be inductively defined as follows.

For every \uitlpm\/ subformula $\psi$ of $\phi$, we define rankers $\mathit{LIntv}(\psi)$ and $\mathit{RIntv}(\psi)$, such that Lemma \ref{lem:uitlpm} holds. $\mathit{LIntv}(\psi)$ and $\mathit{RIntv}(\psi)$ are rankers whose characteristic positions define end points of $\intv(\psi)$ respectively.

\begin{definition}[Composition]
 Let $RK$ be a ranker and $\phi$ be a formula of $\tlxy$. Then, $RK;\phi$ denotes formula $RK[\top/\phi]$. For example,
 $EP (Y_a X_1 \top) ; \phi ~~=~~EP (Y_a X_1 \phi)$.
\end{definition}
Note that if $\phi$ is a ranker then $RK;\phi$ is also a ranker.

\begin{lemma}
\label{lem:uitlpm}
Given a \uitlpm\/ subterm $\psi$ of a formula $\phi$, and any $w\in\alphbt^+$ such that $\intv(\psi),\cpos(\psi)\neq\bot$, 
\begin{itemize}
 \item $\lpos(\mathit{LIntv}(\psi)) = l(\intv(\psi))$
 \item $\lpos(\mathit{RIntv}(\psi)) = r(\intv(\psi))$
\end{itemize}
\end{lemma}
The required formulas $\mathit{LIntv}(\psi), \mathit{RIntv}(\psi)$ may be constructed by induction on the depth of occurrence of the subformula $\psi$ as below. The correctness of these formulas  is apparent from the semantics of \uitlpm\/ formulas, and we omit the detailed proof (see \cite{SShah12}).
\begin{itemize}
 \item If $\psi=\phi$, then $\mathit{LIntv}(\psi) = SP \top$, $\mathit{RIntv}(\psi) = EP \top$
\item If $\psi=SP ~D_1$ then
$\mathit{LIntv}(D_1)=\mathit{RIntv}(D_1)=\mathit{LIntv}(\psi)$
\item If $\psi=EP ~D_1$ then
$\mathit{LIntv}(D_1)=\mathit{RIntv}(D_1)=\mathit{RIntv}(\psi)$
 \item If $\psi = D_1\firsta D_2$ then\\
 $\mathit{LIntv}(D_1) = \mathit{LIntv}(\psi) $, $\mathit{RIntv}(D_1) = \mathit{LIntv}(\psi)~;~ \weakx{a}\top$,\\
 $\mathit{LIntv}(D_2) = \mathit{LIntv}(\psi) ~;~ \weakx{a}\top$, $\mathit{RIntv}(D_2) = \mathit{RIntv}(\psi)$
\item If $\psi=D_1\firstp{a}D_2$ then\\
 $\mathit{LIntv}(D_1) = \mathit{LIntv}(\psi) $, $\mathit{RIntv}(D_1) = \mathit{RIntv}(\psi)~;~ \weakx{a}\top$,\\
 $\mathit{LIntv}(D_2) =\mathit{RIntv}(\psi) $, $\mathit{RIntv}(D_2) = \mathit{RIntv}(\psi) ~;~ \weakx{a}\top$
\item If $\psi=D_1\lasta D_2$ then\\
$\mathit{LIntv}(D_1) = \mathit{LIntv}(\psi) $, $\mathit{RIntv}(D_1) = \mathit{RIntv}(\psi)~;~ \weaky{a}\top$,\\
 $\mathit{LIntv}(D_2) = \mathit{RIntv}(\psi) ~;~ \weaky{a}\top$, $\mathit{RIntv}(D_2) = \mathit{RIntv}(\psi)$
\item If $\psi=D_1\lastm{a}D_2$ then\\
$\mathit{LIntv}(D_1) = \mathit{LIntv}(\psi) ~;~ \weaky{a}\top $, $\mathit{RIntv}(D_1) = \mathit{LIntv}(\psi)$,\\
 $\mathit{LIntv}(D_2) = \mathit{LIntv}(\psi) ~;~ \weaky{a}\top$, $\mathit{RIntv}(D_2) = \mathit{RIntv}(\psi)$
\item If $\psi = \succr D_1$ then 
$\mathit{LIntv}(D_1) = \mathit{LIntv}(\psi)~;~ \xunit\top$, $\mathit{RIntv}(D_1)=\mathit{RIntv}(\psi)$
\item If $\psi = \succp D_1$ then 
$\mathit{LIntv}(D_1) = \mathit{LIntv}(\psi)$, $\mathit{RIntv}(D_1)=\mathit{RIntv}(\psi) ~;~ \xunit\top$
\item If $\psi = \predr D_1$ then 
$\mathit{LIntv}(D_1) = \mathit{LIntv}(\psi)$, $\mathit{RIntv}(D_1)=\mathit{RIntv}(\psi) ~;~ \yunit\top$
\item If $\psi = \predm D_1$ then 
$\mathit{LIntv}(D_1) = \mathit{LIntv}(\psi) ~;~ \yunit\top$, $\mathit{RIntv}(D_1)=\mathit{RIntv}(\psi) $
\end{itemize}

Now we can give a model preserving transformation.
\begin{theorem} \label{theo:uitlpm}
Given any \uitlpm\/ formula $\phi$ of size $n$, we can construct in polynomial time a language-equivalent \tlxy\/ formula
$\mathit{Trans}(\phi)$, whose size is $O(n^2)$. Hence, satisfiability of 
\uitlpm\/ is NP-complete. 
\end{theorem}
\begin{proof}
For any subformula $\psi$ of $\phi$, we construct a corresponding \tlxy\/ formula $\mathit{Trans}(\psi)$. 
The conversion uses the following inductive rules.  Then, it is easy to see that $\mathit{Trans}(\psi)$ is language equivalent to $\phi$
(see \cite{SShah12} for proof). 
\begin{itemize}
\item If $\psi= SP ~D_1$ or $EP ~D_1$ then $\mathit{Trans}(\psi) = LIntv(D_1)~;~\mathit{Trans}(D_1)$
 \item If $\psi = D_1\firsta D_2$, then $\mathit{Trans}(\psi)= [(~\mathit{LIntv}(\psi);\weakx{a}\top~)~;~\pfleq(\mathit{RIntv}(\psi))] \land \mathit{Trans}(D_1)\land \mathit{Trans}(D_2)$
 \item If $\psi = D_1\lasta D_2$, then $\mathit{Trans}(\psi)= [(~\mathit{RIntv}(\psi);\weaky{a}\top~)~;~\pfgeq(\mathit{LIntv}(\psi))] \land \mathit{Trans}(D_1)\land \mathit{Trans}(D_2)$
 \item If $\psi = D_1\firstp{a} D_2$, then $\mathit{Trans}(\psi)= [(~\mathit{LIntv}(\psi);\weakx{a}\top~)~;~\pfgeq(\mathit{RIntv}(\psi))] \land \mathit{Trans}(D_1)\land \mathit{Trans}(D_2)$
 \item If $\psi = D_1\lastm{a} D_2$, then $\mathit{Trans}(\psi)= [(~\mathit{RIntv}(\psi);\weaky{a}\top~)~;~\pfleq(\mathit{LIntv}(\psi))] \land \mathit{Trans}(D_1)\land \mathit{Trans}(D_2)$
\item If $\psi= \succr D_1$, then $\mathit{Trans}(\psi)=[(\mathit{LIntv}(\psi);\xunit\top) ~;~ \pfleq(\mathit{RIntv}(\psi))] ~\land~ \mathit{Trans}(D_1)$
\item If $\psi= \predr D_1$, then $\mathit{Trans}(\psi)=[(\mathit{RIntv}(\psi);\yunit\top) ~;~ \pfgeq(\mathit{LIntv}(\psi))] ~\land~ \mathit{Trans}(D_1)$
\item If $\psi= \succp D_1$, then $\mathit{Trans}(\psi)=[(\mathit{RIntv}(\psi);\xunit\top)] ~\land~ \mathit{Trans}(D_1)$
\item If $\psi= \predm D_1$, then $\mathit{Trans}(\psi)=[(\mathit{LIntv}(\psi);\yunit\top)] ~\land~ \mathit{Trans}(D_1)$
 \item $\mathit{Trans}(D_1\lor D_2) = \mathit{Trans}(D_1)\lor \mathit{Trans}(D_2)$
 \item $\mathit{Trans}(\neg D_1) = \neg \mathit{Trans}(D_1)$
\end{itemize}
\end{proof}
\subsection{AtNext Logic}

Logic $\tlxy$ exactly characterizes the language class $UL$. The previous section shows that several deterministic logics can be translated
to $\tlxy$ in polynomial time using the rankers and ranker directionality. Thus, there is robust connection between $UL$, deterministic modalities and efficient NP-complete satisfiability. 

In this section, we consider a recursive (hierarchical) extension of $\tlxy$ which is deterministic but much more expressive, 
Recursive Temporal Logic (\rectl\/) with the \emph{recursive} and \emph{deterministically guarded} \emph{Next} and \emph{Prev} modalities. 
The logic \rectl\/ was defined by Kr\"oger \cite{Kro}, with ``at-next'' and ``at-prev'' modalities and shown to be expressively
equivalent to $\ltl$.
\begin{definition}[Syntax]
$\phi := \top ~\mid~  a ~\mid~ X_\phi \phi~\mid~ Y_\phi \phi~\mid~ \phi\lor\phi ~\mid~ \neg \phi$
\end{definition}
When interpreted over a word $w$ and at a position $i$ in $w$, 
the semantics of the X and Y operators is given by:
\begin{itemize}
\item $w,i\models X_\phi\psi$ iff $\exists j>i~.~ w,j\models\phi\land\psi$ and $\forall i<k<j ~.~ w,k\not\models \phi$
\item $w,i\models Y_\phi\psi$ iff $\exists j<i~.~ w,j\models\phi\land\psi$ and $\forall j<k<i ~.~ w,k\not\models \phi$
\end{itemize}
Given a \rectl\/ formula $\phi$, we may define the recursion depth $rd(\phi)$ using the following rules:
\begin{itemize}
\item If $\phi = a$ or $\phi=\top$, $rd(\phi)=0$.
\item If $\phi = \phi_1\lor\phi_2$, $rd(\phi) = max(rd(\phi_1),rd(\phi_2))$
\item If $\phi= \neg\phi_1$, $rd(\phi)=rd(\phi_1)$
\item If $\phi = X_\zeta\psi$ or $\phi = Y_\zeta\psi$, then  $rd(\phi) ~=~ max(rd(\zeta)+1, rd(\psi))$
\end{itemize}
We denote by $\rectl^k$ formulae with maximum recursion depth $k$. Note that there is no restriction on nesting depth of modalities.
It is clear that $\rectl^1 ~=~ \tlxy$.

\begin{example}
\label{exam:rectl}
Consider the \tlrecr\/ formula $\phi =X_{\psi_1}Y_{\psi_2}\top$ where $\psi_1= a\land Y_b\top\land X_c\top$ and $\psi_2= X_c \habar{b}$. When we evaluate $\phi$ over the word $w= ccaccbccabbcacc$, $\pos(\phi)=1$. The first position in the word where $\psi_1$ holds is 9 hence $\pos(Y_{\psi_2}\top)=9$. Finally, the last position before 9 where $\psi_2$ holds is 4. Hence $w\in\mathcal L(\phi)$.
\end{example}

A closer look at the semantics of \rectl\/ and \ltl\/ allows us to see that the deterministic until and since modalities are in fact not very different from the \emph{until} ($\until$) and \emph{since} ($\since$) modalities of LTL. Translations between them may be achieved using translation functions $\alpha$ and $\beta$ as described below. 
\begin{lemma} 
$TL[\until,\since]^k ~\leq~ \rectl^k$. Hence, $\rectl ~\equiv~ \ltl~ \equiv FO[<]$.
\end{lemma}
\begin{proof}
Let the translation functions which preserve boolean operations be defined as follows.
\begin{itemize}
\item $\alpha(\phi \until \psi) ~~\equiv~~ X_
                 {\alpha[(\neg \phi) \lor \psi]} ~~\alpha(\psi)$
\item $\beta(X_\phi \psi) ~~\equiv~~ 
                 [\beta(\neg\phi)] \until [\beta(\phi\land\psi)]$
\end{itemize}
The \emph{since} modalities may be translated in a similar manner. Then, it is easy to show by induction on the
depth of formulae that
\begin{itemize}
 \item for any \ltl\/ formula,
$w,i \models \phi ~~\fif~~ w,i \models \alpha(\phi)$.
 \item for any $\rectl$ formula $\psi$,
$w,i \models \psi ~~ \fif ~~w,i \models \beta(\psi)$.
\end{itemize}
 Note that for $\phi \in \ltl^k$  we have $\alpha(\phi) \in \rectl^k$.
 Also, note that it is straightforward to translate $\rectl$ formulae into $FO[<]$ formulae with one free variable $x$.
\end{proof}
\\
We remark here that Simoni Shah has recently come up with a  form of alternating automata called \rpotdfa\/ such that
$\rectl^k$ exactly corresponds to $\rpotdfa^k$. Thus, there is a clean automaton characterization for the AtNext hierarchy.
We also remark that languages $Stair_k= \alphbt^*(a c^*)^ka\alphbt^*$ 
defined by Etessami and Wilke \cite{EW} are specified by
$\rectl^2$ formula $X_\psi$ where $\psi = \neg X_{a \lor c} a \land X_{a \lor c} a \land \ldots X_{a \lor c} a$ with $k$ occurrences of $X$. See \cite{PS15} for details.

Finally, the following theorem relates the At-Next Hierarchy to the Quantifier-Alternation Hierarchy of Thomas \cite{Tho}.
\begin{theorem}[Borchert and Tesson \cite{BT04,PS15}]
$\rectl^k ~\subseteq~ \Delta_{k+1}[<]$.
\end{theorem}
We do not give a proof of this here. See \cite{BT04} for a proof outline. 
Explicit translations from $\rectl^k$ to formulae of $\Sigma_{k+1}[<]$ as well as $\Pi_{k+1}[<]$ are given in \cite{PS15}.

Now we consider a subset of $\rectl$ called \tlrecr.
\begin{definition}[Syntax]
$\psi ~:=~ a ~\mid~ \phi ~\mid~ \psi\lor\psi ~\mid~ \neg\psi$, 
where $a \in\alphbt$ and

\hspace{3cm}$\phi ~:= \top ~\mid~  SP\phi ~\mid~ EP\phi ~\mid~ X_\psi \phi ~\mid~ Y_\psi \phi$ 
\end{definition}
The formula in example \ref{exam:rectl}  is actually a formula of $\tlrecr$.

In the above syntax the $\phi$ formulae are called the 
\emph{recursive rankers} of \tlrecr. 
The main restriction is that rankers cannot use boolean operators 
or the atomic proposition $a$ except through recursive subformulae.  
The recursive rankers satisfy an important property
of convexity as stated below. 
\begin{lemma}[Convexity \cite{PS14}]
\label{lem:convex2}
For any recursive ranker formula $\phi$, and any word $w\in\alphbt^+$, if there exist $i,j\in dom(w)$ such 
that $i<j$ and $w,i\models\phi$ and $w,j\models\phi$, then $\forall i<k<j$, we have $w,k\models\phi$. 
\end{lemma}

Simoni Shah \cite{SShah12,PS14} has shown the following result.  
\begin{lemma}
$ \tlrecr ~~\equiv~~ \tlxy$
\end{lemma}
\begin{proof}
Any $\tlxy$ formula can be syntactically normalized to equivalent boolean combination of rankers, and hence
$\tlxy \subseteq \tlrecr$. 

For the converse, we only give a reduction from $\tlrecr$ to $\utl$ which is known to be expressively equivalent to
$\tlxy$ \cite{EVW02,DGK08}.
For any $\psi \in \tlrecr$, we will construct \utl\/ formulas $\at(\psi)$ such that 
$\forall w\in\alphbt^+$ we have $w,i\models \at(\psi)$ iff $w,i \models \psi$. The construction
is by induction on the structure of $\psi$ (and its rankers $\phi$). 
Define $\at(a) = a$, $\at(\top)=\top$ and $\at(\mathcal B(\phi_1, \ldots \phi_m)) = {\mathcal B(\at(\phi_1), \ldots \at(\phi_m))}$. 
It is easy to see that 
$w,j \models \at(\mathcal B(\phi_1, \ldots \phi_m))$ iff $w,j \models  \mathcal B(\phi_1, \ldots \phi_m)$. 

Now, we give the reduction for recursive ranker formulae $\phi$. 
The figure 
below (from \cite{PS14}) depicts convexity of $\phi=X_\psi \phi_2$.
It shows the positions where $\phi$ holds in a word $w$. 
Note that there is only one convex interval where $\phi_2$ holds. 
$\phi$ holds at positions where 
$F(\phi_2 \land \psi)$ is true but no future position
has $\psi \land \neg \phi_2 \land F\phi_2$. Thus:
\begin{math}
\begin{array}{l}
\at(X_{\psi_1}(\phi_2)) = ~\fut[\at(\psi_1)\land \at(\phi_2)]  ~\land~ 
		\neg\fut[\at(\psi_1)\land \neg \at(\phi_2) \land F \at(\phi_2))], \\
\at(Y_{\psi_1}(\phi_2)) = ~\past[\at(\psi_1)\land \at(\phi_2)]  ~\land~ 
		 \neg\past[\at(\psi_1)\land \neg \at(\phi_2) \land \past \at(\phi_2))]. 		
\end{array}
\end{math}
\end{proof}
\begin{figure}[h]
\begin{center}
\begin{tikzpicture}[scale=1,transform shape]
\draw (5,4.6) node{$\phi = X_{\psi_1}\phi_2$};
\draw[thick] (3,4.5) node{[} -- (7,4.5) node{]};
\draw (0,4) node{$w$}; \draw (0.2,4) node{l}-- (1,4) node{l} -- (2,4) node{l}--(3,4) node{l}--(4,4) node{l}--(5,4) node{l}--(6,4) node{l}--(7,4) node{l}--(8,4) node{l}--(9,4) node{l}--(10,4) node{l};
\draw (1,3.5) node{$\psi_1$}; \draw (3,3.5) node{$\psi_1$}; \draw (5,3.5) node{$\psi_1$}; \draw (6,3.5) node{$\psi_1$}; \draw (8,3.5) node{$\psi_1$}; \draw(10,3.5) node{$\psi_1$};
\draw (5,3.1) node{[}--(8,3.1) node{]};
\draw (6.5,3) node{$\phi_2$};
\end{tikzpicture}
\end{center}
\label{fig:tlrecr}
\end{figure}

\section{Interval Constraints}
\label{sec:ic}
In the previous section we studied various deterministic and deterministically
guarded temporal logics.
It was seen that these allow efficient algorithms and decision procedures 
compared to full \ltl. In this section we retain the ``unary" flavour of
these logics but we expand its scope to gain expressiveness. 
More precisely we consider a unary temporal logic \blintl\/ where 
the binary Until and Since modalities of \ltl\/ are guarded by 
interval constraints on the left, allowing counting or simple algebraic 
operations, forming guarded unary operators $g\until\phi$ and $g\since\phi$. 
The techniques are borrowed from full \ltl, and the complexity of 
decision procedures jumps to that of \ltl. In fact it is one exponent 
more when using binary notation (as is also the case for \ltl). 
The analogues of ``rankers" or ``turtle programs" 
remain to be discovered in this setting.

We do not have precise expressiveness results for most of these logics. 
What is perhaps surprising is that \blintl, even though unary, 
has formulae which reach all levels of the Until/Since hierarchy for \ltl\/
of Th\'erien and Wilke \cite{TW-until} as well as the dot depth hierarchy for
starfree expressions of Cohen, Brzozowski and Knast \cite{CB,BK} and the 
quantifier alternation hierarchy for first-order logic of Thomas \cite{Tho}. 
Thus it is quite an expressive, yet succinct logic. Emerson and Trefler 
argued for introducing counting in binary into temporal logic using 
a starfree expression syntax \cite{ET}. \blintl\/ is moreover elementarily 
decidable, a line of work we have been following \cite{LPS08,LPS,LS,KLPS}.
 
We define some families of constraints below:
respectively, \name{simple}, \name{modulo counting}, \name{group counting},
\name{threshold}, \name{linear} and \name{ordered group} 
constraints, for some of which we consider boolean closure also. 
The constraints with groups are generalizations of those dealing with integers.

Let $B,B_i \subseteq A$, let $c_i\in \Int$, $t,u \in \Nat$. 
For $q \in \Nat \setminus \{0,1\}$, we write $[q]$ for $\{0,\dots,q-1\}$.
We also let $G = \{h_1,\dots,h_k,0\}$ 
be a finite group written additively, 
with its elements enumerated in a linear order.
We assume that the name $G$ identifies the group and this ordering, 
its description does not enter our syntax.
We also consider a finitely generated discretely ordered abelian group $O$ 
with $l,m \in O$,
such that $F = \{g_1,\dots,g_k,0\}$ fixes a linear order over its generators 
and the identity element. 
\[
\begin{array}{lll}
   sg & ::= & \#B = 0  \\
   modg & ::= & \sumover{i} c_i\#B_i \in R \mod q,~\rmwhere R \subseteq [q]\\
   grpg & ::= & \sumover{G}(c_1\#B_1,\dots,c_k\#B_k)\in H,~\rmwhere H\subseteq G\\
   thrg & ::= & t \oprt \#B ~\mid~ \#B \oprt u ~\mid~
	t \oprt \#B \oprt' u,~\rmwhere\oprt,\oprt'\rmin\{<,\leq\}\\
   ling & ::= & modg ~\mid~ thrg\\
   ogpg & ::= & l \oprt \sumover{O}(c_1\#B_1,\dots,c_k\#B_k) \oprt' m\\
   bsg & ::= &  sg ~\mid~ bsg_1 \land bsg_2 ~\mid~ \neg bsg ~\mid bsg_1 \lor bsg_2\\
   btg & ::= &  thrg ~\mid~ btg_1 \land btg_2 ~\mid~ \neg btg ~\mid btg_1 \lor btg_2\\
   bg & ::= & ling ~\mid~ bg_1 \land bg_2 ~\mid~ \neg bg ~\mid bg_1 \lor bg_2
\end{array}
\]
For a modulo counting constaint, if $R$ is a singleton $\{r\}$ we write 
$\sumover{i}c_i\#B_i\equiv r$.
For a threshold counting constraint, if $t=u$ we write $\#B=t$.

Given a word $w \in A^+$ and $x,y \in dom(w)$, let $\#B(w,x,y)$ denote 
the number of occurrences of letters in $B$ positions $x$ to $y$ inclusive. 
Also, given group $G$, define $G(w,z) = c_j h_j$ 
(and $O(w,z) = c_j g_j$, respectively) 
if $w[z] \in B_j \setminus (B_1 \cup \dots \cup B_{j-1})$
for $1 \leq j \leq k$, and otherwise zero (the identity element) if
$w[z] \notin (B_1 \cup \dots \cup B_k)$. 
We say:
\[
\begin{array}{l}
w,[x,y] \models \sumover{i}c_i\#B_i \in R \mod q ~\fif~
\sumover{i} c_i\#B_i(w,x+1,y-1) \in R \mod q\\ 
w,[x,y] \models \sumover{G}(c_1\#B_1,\dots,c_k\#B_k) \in H ~\fif~
\sumbounds{z=x+1}{y-1} G(w,z) \in H\\
w,[x,y] \models t \oprt \#B \oprt' u ~\fif t \oprt \#B(w,x+1,y-1) \oprt' u\\
w,[x,y] \models l \oprt \sumover{O}(c_1\#B_1,\dots,c_k\#B_k) \oprt' m ~\fif
	l \oprt \sumbounds{z=x+1}{y-1} O(w,z) \oprt' m
\end{array}
\]
This can be extended to boolean guards as usual.

Our logic \blintl\/ over $A$ has the following syntax, where the 
Until and Since ($\until,\since$) modalities of \ltl\/ are used in 
a unary fashion.
\[
\phi ::=   a ~\mid~ \neg \phi ~\mid~ \phi \lor \phi ~\mid~ bg\until\phi ~\mid~ bg\since\phi
\]
Given a word $w \in A^+$ and position $i \in dom(w)$, 
the semantics of a \blintl\/ formula is given below. 
Boolean operators have the usual meaning.
The same definitions would work for infinite words, which are
more usual as models for temporal logics.
\[
\begin{array}{l}
w,i \models a  ~\fif~ w[i]=a \\
w,i \models bg\until\phi ~\fif~ \exists j>i \such w,[i,j] \models bg \rmand w,j \models \phi \\
w,i \models bg\since\phi ~\fif~ \exists j<i \such w,[j,i] \models bg \rmand w,j \models \phi
\end{array}
\]
\paragraph{Size}
Size $|\phi|$ of a formula $\phi$ and modal depth are defined as usual.
Constants are encoded in binary and size of a set of letters $B$ is the number of elements in $B$.
Thus, $|(\neg (\#\{b,c\} > 1) \land \#a =17)\until a|$ is 
$max(|\neg \#\{b,c\} > 1|,~|\#a =17|)+1$, which works out to
$max(2,\lceil \log_2 17 \rceil)+1=6$.

\paragraph{Abbreviations}
We shall use the abbreviation $B\until \phi$ for $(\#(A-B) = 0)\until \phi$.
Also, $\fut \phi = A\until \phi = \true \until \phi$,
$\henceforth\phi = \lnot\fut\lnot\phi$,
$\nextt \phi = \emptyset\until \phi = \false\until\phi$. 
For a guard $g$, the formula $\now g = g\since(\neg\prev\true)$ gives the
current velaue of a guard evaluated from the first position of the word.
For initializing and updating guards, we use:
\begin{itemize}
\item If $g$ is $\sumover{i} c_i\#B_i \equiv r\mod q$, then
$g(0)$ is $\sumover{i} c_i\#B_i \equiv 0\mod q$ and
$g+c_j$ is $\sumover{i} c_i\#B_i \equiv r+c_j\mod q$.
\item If $g$ is $\sumover{G}(c_1\#B_1,\dots,c_k\#B_k)=h$, then
the guard $g(0)$ is $\sumover{G}(c_1\#B_1,\dots,c_k\#B_k)=0$ and
the guard $g+h'$ (we will use $h'=c_j h_j$ below) is 
$\sumover{G}(c_1\#B_1,\dots,c_k\#B_k)=h+h'$.
\item If $g$ is $\#B\oprt v$ and $a \in A$, then
$g-a$ is $\#B\oprt v-1$ if $a \in B$, and $g$ otherwise.
\end{itemize}

\paragraph{Sublogics}
\begin{itemize}
\item Logic \btlth\/ is a subset of \blintl\/ where modalities use only
threshold constraints $btg\until\phi$ and $btg\since\phi$. 
\item Logic \btlinv\/ is  subset of \btlth\/ where modalities use boolean
combinations of simple constraints $bsg\until\phi$ and $bsg\since\phi$. 
\item Logic \utlinv\/ is  subset of \btlinv\/ where modalities use only
simple constraints $sg\until\phi$ and $sg\since\phi$. 
\item Logic \utlmodinv\/ is a subset of \blintl\/ where modalities use
only simple and modulo counting constraints 
$sg\until\phi,modg\until\phi$ and $sg\since\phi,modg\since\phi$.
\end{itemize}

We note that Unary \ltlunsuc\/ is a subset of \btlth.
The guard $\#A=u$ expresses the $u+1$-iterated Next operator $\nextt^u$.
Since $n$ is written in binary this gives an exponential succinctness
to this logic over \ltlunsuc.

\paragraph{Examples}
The formula 
$(\#b\equiv 1\mod 3)\until((\#a\equiv 0\mod 2)\until\neg\nextt\true)$
says that every word has $3n+1$ occurrences of the letter $b$, 
for some $n \geq 0$, followed by an even number of occurrences of the 
letter $a$, excluding the last letter on the word. Such modulo counting 
is not expressible in \ltl\/ or first-order logic \cite{Wolper}.
Notice that the syntax allows nesting $\until$ (and $\since$) modalities
on the right but not on the left. 

Several interesting languages not in \ltlunsuc\/ 
can be specified in \btlth. 

The language $Stair_k$ which specifies $k$ occurrences of the letter $a$ 
without any intermediate occurrences of letter $b$ \cite{EW} is specified
by the formula $\fut(a \land (\#b=0 \land \#a=k-2)\until a)$.

The formula
$\henceforth( b \land (c\until b) \limplies \past(a \land (c\since a)))$
defines language $U_2$ in \utlinv\/ (a simpler version appears in \cite{LPS}), 
which specifies over a 3-letter alphabet that if a word has an occurrence of
two $b$'s without an $a$ between them, then it must be preceded by an 
occurrence of two $a$'s without a $b$ between them. 

One can also define expressively equivalent two-variable fragments 
of first-order logic corresponding to the classes of interval constraints, 
as in our earlier work \cite{KLPS}, but we do not pursue this here.

\subsection{Expressiveness}

Given logics $L_1$ and $L_2$ over finite words, we can relate them by their expressive powers. We use $L_1 \subseteq L_2$ if for $\forall \phi \in L_1 \exists
\psi \in L_2 \st (w \models \phi \fif w \models \psi)$. 
We use $L_1 \equiv L_2$ if $ L_1 \subseteq L_2$ and $L_2 \subseteq L_1$.
The next two theorems show that Boolean operations over threshold and 
modulo counting constraints can be eliminated. In fact threshold counting 
can be reduced to invariant counting or to modulo counting. In our 
earlier paper \cite{KLPS}, we used the first reduction as the basis 
for a decision procedure. Here we use the second theorem as the basis
for our decision procedure.

\begin{theorem}
\label{theo:bthtlred}
$\btlth ~\equiv~ \utlinv$.
\end{theorem}
\begin{proof}
Threshold constraints $\#B \geq 0$ and $\#\emptyset = 0$ 
can be replaced by $\true$, $\#A=0$ by $\false$. 
Also multiple upper bounds and
lower bounds on the same set of letters can be combined, for example
replacing $(\#B \geq t_1 \land \#B \geq t_2)$ by $\#B \geq max(t_1,t_2)$.
We also remove obviously contradictory conjunctions.
To eliminate negations, we have:
\[\neg (\#B \geq t ) \equiv \#B < t,~ 
\neg (\#B \leq u ) \equiv \#B > u,~
\neg (t \leq \#B \leq u) \equiv (\#B < t) \lor (\#B > u).
\]
Only one of the disjuncts above can hold for all prefixes, 
since the count of a letter cannot jump from below $t$ to above $u$. 
(If $t > u$ comes from another conjunct inside the negation, 
both disjuncts hold since we have a tautology.)
This generalizes for 
a non-tautological disjunction of threshold constraints $\eta_1,\eta_2$ to:
\[
   (\eta_1 \lor \eta_2)\until \phi ~\equiv~ 
       (\eta_1 \until \phi) ~\lor~  (\eta_2\until \phi) 
\]
We will not specify mirror image rules for the past modalities here and below.
As usual, the boolean conditions can be put in disjunctive normal form.
Applying these rules we can obtain  $bg \equiv \orover{} CN$ 
where $CN$ is conjunction of simple threshold constraints $g$. 

Finally, let us consider a \btlth\/ constraint of the form $\#B \leq u$
which has been brought to this form. 
This can be replaced by $\#B = 0 \lor \ldots \lor \#B = u$, 
and the disjunctions can be moved outside the modalities.

With all this, we obtain $bg\until \psi \equiv \orover{} (NCN\until \psi)$ 
where $NCN$ is a conjunction of equality and lower bound constraints, 
each set of letters $B$ occurring in at most one constraint. 
$NCN = AC \cup BC \cup CC$ where $AC$ are constraints of the form $\#B =0$, 
$BC$ are constraints of the form $\#B = c$ with $c>0$ and 
$CC$  are constraints of the form $\#B \geq c$ with $c>0$. 
We have:
\[
(AC \cup BC \cup CC)\until \psi ~\equiv~ 
\orover{a \in (BC \cup CC)-AC} (AC=0 \cup BC=0 \cup CC=0) 
\until(a \land (AC \cup BC-a \cup CC-a)\until \psi)
\]
By repeated application of the above rule, 
we can get an equivalent formula where all
the constraints $AC$ are conjunctions of the form
$\#B_i = 0$. This is equivalent to a single constraint $\#\cup B_i = 0$. 
A similar reduction can be carried out for the past modalities. Hence,
$\btlth \subseteq \utlinv$. 

Starting with a \btlth\/ formula, the reduction gives rise to 
an exponential (in product of constant $c$ and alphabet size $m$) 
blowup in modal depth of the formula, 
since updating by an occurrence of $a$ changes $\#B=c+1$ to $\#B=c$
for $a \in B$.
Starting with \btlinv, the modal depth increases by one for each letter of alphabet,
since updating by an occurrence of $a$ in $B$ changes 
$\#B>0$ and $\#B=0$ to $\true$. Hence modal depth blows up by 
the size of the alphabet.
\end{proof}

\begin{corollary}[\cite{KLPS}]
\label{coro:bthtl}
The satisfiability of \btlth\/ is complete for \Expspace.
\end{corollary}
\begin{proof}
The transation above gives an exponential-sized formula, which is 
easily translated into the syntax of \ltl. By the decision procedure 
for \ltl\/ \cite{SC}, this gives an \Expspace\/ upper bound.
Since the Counting Next and Future modalities ($\nextt^u,\fut$)
of $\ltl$ (with $u$ in binary) are definable,
\Expspace\/ is also a lower bound \cite{AH,ET}.
\end{proof}

\begin{theorem}
\label{theo:btlthredmod}
$\blintl ~\equiv~ \utlmodinv$.
\end{theorem}
\begin{proof}
Modulo and threshold counting requirements for guarded Until formulas 
can be reduced to checking \name{global counters}---constraint values
$\now g$ counted ``from the beginning''. 
This is shown in the tautologies below.
The first line reduces modulo counting formulae, after that we reduce 
threshold to modulo counting.
The right hand formulae below have size multiplied by a factor of $q$ or $u$,
so they are exponential in the binary representation of 
$q,u\geq 2$ or $t\geq 1$. 
\[\begin{array}{l}
(\sumover{i}c_i\#B_i \equiv r\mod q)\until\phi \eqvt \orover{r_0 \in [q]} 
(\now \sumover{i}c_i\#B_i \equiv r_0\mod q)\land
\fut(\phi \land (\now \sumover{i}c_i\#B_i \equiv r_0+r\mod q)\\
(\#B<u)\until\phi \eqvt \orover{r_0 \in [u]} 
(\now \#B\equiv r_0\mod u)\land\lnot(\now \#B\equiv r_0\mod u)\until\phi\\
(\#B=u-1)\until\phi \eqvt\\
\hspace{2cm}\orover{r_0 \in [u]}(\now \#B\equiv r_0\mod u)
	\land
	\lnot(\now \#B\equiv r_0\mod u)
	\until(\phi \land(\now \#B\equiv r_0-1\mod u))\\
(t\leq\#B)\until\phi \eqvt (\#B=t)\until(\fut\phi)\\
(t\leq\#B<u)\until\phi\eqvt\orover{r_0 \in [u]}(\now \#B\equiv r_0\mod u)\land\\
\hspace{2cm}\lnot(\now \#B\equiv r_0+t\mod u)\until
		((\now \#B\equiv r_0+t\mod u)\land
		\lnot(\now \#B\equiv r_0\mod u)\until\phi)
\end{array}
\]
Now observe that for modulo (and group) counting constraints
one can perform the boolean operation on the specified elements $R$ or $H$.
For different moduli, we have to take least common multiples of the quotients
leading to a polynomially larger formula.
(For different groups, we have to take products.)
\end{proof}

In each case above, only one of the right hand disjuncts can hold. 
At a given point in a model, it is possible that both 
$(\#a=10)\until\phi$ and $(\#a=5)\until\phi$ hold, but the value $r$ of the 
global $a$-counter $\now \#a\equiv r\mod u$ is unique. 
We will use this below.

\subsection{Subformulas and the formula automaton}
Fix a formula $\alpha_0$. 
The \name{Fischer-Ladner closure} of a formula $\alpha_0$ \cite{FL} 
is constructed as usual, some of the clauses below are 
based on the global counter tautologies in Theorem~\ref{theo:btlthredmod}.

\begin{enumerate}
\item $\alpha_0$ is in the closure.
\item If $\phi$ is in the closure, $\lnot\phi$ is in the closure.
We identify $\lnot\lnot\phi$ with $\phi$.
\item If $\phi \lor \psi$, $\phi\until\psi$ and 
$\phi\since\psi$ are in the closure, so are $\phi$ and $\psi$.
\item The closure of a set with 
$(\sumover{i} c_i\#B_i \in R \mod q)\until\phi$ includes:\\
$\fut(\phi\land \now \sumover{i} c_i\#B_i \equiv r \mod q)$ 
and $\now \sumover{i} c_i\#B_i \equiv r \mod q$,
for every $r$ in $[q]$.
\item The closure of a set with 
$(\sumover{G}(c_1\#B_1,\dots,c_k\#B_k) \in H)\until\phi$ includes:\\
$\fut(\phi\land \now \sumover{G}(c_1\#B_1,\dots,c_k\#B_k)=h)$ 
and $\now \sumover{G}(c_1\#B_1,\dots,c_k\#B_k)=h$, 
for every $h$ in $G$.
\item The closure of a set with $(t \leq \#B < u)\until\phi$ includes:\\
$(\lnot(\now \#B\equiv r\mod u)\until
(\now \#B\equiv r\mod u)\land
	\fut(\phi\land \now \#B\equiv s\mod u)$,\\
$\fut(\phi\land \now \#B\equiv s\mod u)$ and
$\now \#B\equiv s\mod u$, for every $r$ and $s$ in $[q]$.
%
\end{enumerate}

Unlike the usual linear size for LTL, since
the constants $c_i$, $R,r,s,q$, $H,h$, $t,u$, $l,o,m$ 
are written in binary notation,
the closure of a modulo or group counting formula $\alpha_0$ 
is exponential in the size of $\alpha_0$. In the case of a
threshold formula the closure is $O(2^{|\alpha_0|^2})$. 

A \name{state} (sometimes called an \name{atom}) 
is a maximal Hintikka set of formulae 
from the Fischer-Ladner closure of $\alpha_0$. 
Gabbay, Hodkinson and Reynolds \cite{GHR} use the more classical notion of a 
\name{$k$-type} (Hintikka set with formulas upto modal depth $k$).
Assume an enumeration of formulae in the state. 
Instead of using an explicit indexing, 
we loosely use the formula $\phi$ as though it uniquely identifies 
a particular formula of the form $g \until \phi$ or $g \since \phi$. 

For every constraint,
only one of the exponentially many global counter formulae with 
modulus value $r$ in $[q]$,
or with group element value $h$ in $G$,
can hold in a state.
Hence the number of states, although it has a subset of the closure of 
$\alpha_0$ which is already exponential in the size of $\alpha_0$, 
grows only exponentially with the size of $\alpha_0$ even though 
the modulo and group counting constants are represented in binary 
\cite{LS,Sree}.
So a state can be represented using space polynomial in the size of $\alpha_0$.

Next we define a \name{transition} relation from state $s_1$ to state $s_2$.
Suppose $g\until\phi$ is in $s_2$, we specify the requirements on $s_1$,
and if it is in $s_1$, then the requirements on $s_2$. 
Looking at the requirements below, it is easy to derive 
the mirrored requirements for $g\since\phi$.
Assume without loss of generality that
all subalphabets $B_i$ mentioned in the guard $g$ 
are disjoint from each other. 

\begin{enumerate}
\item If $g$ is $\sumover{i} c_i\#B_i \equiv r\mod q$ 
and if $a \in B_j$ for some $j$ is in $s_2$, then:
\begin{itemize}
\item $g\until\phi$ in $s_2$ implies $(g+c_j)\until\phi$ in $s_1$;
\item $g\until\phi$ in $s_1$ implies $(g-c_j)\until\phi$ in $s_2$.
\end{itemize}
\item If $g$ is $\sumover{G}(c_1\#B_1,\dots,c_k\#B_k)=h$ 
and if $a \in B_j$ for some $j$ is in $s_2$, then:
\begin{itemize}
\item $g\until\phi$ in $s_2$ implies $(g+c_j h_j)\until\phi$ in $s_1$;
\item $g\until\phi$ in $s_1$ implies $(g-c_j h_j)\until\phi$ in $s_2$.
\end{itemize}
\item If $g$ is a threshold constraint $\#B=0$ and 
if $a \in B$ is in $s_1$, then
$g\until\phi$ in $s_1$ implies $\phi$ in $s_2$.
\item
In each case above, if the alphabetic precondition is not satisfied,
depending on whether it was assumed to be in $s_2$ or $s_1$,
then $g\until\phi$ is required to be in $s_1$ or $s_2$ respectively.
\item
In each case of modulo constraint $g$ above, 
if $\phi$ is in $s_2$, then $g(0)$ is in $s_1$.
\end{enumerate}

Since there are exponentially many states, each state as well
as the transition relation of an exponential size 
\name{formula automaton} can be represented in polynomial space. 

\subsection{Decision problems}

For the unary \ltlun, Ono and Nakamura \cite{ON} 
use the convexity of the Future and Past ($\fut,\past$) modalities 
to derive that only polynomially many distinct states need appear 
on a path to witness the satisfaction of the modalities, and hence 
that its satisfiability and model checking problems are decidable in \Np. 
Sistla and Clarke, and Lichtenstein and Pnueli \cite{SC,LP85} showed that 
the satisfiability problem for \ltl\/ is in \Pspace\/
(see also the monograph \cite{GHR} for an analysis based on types), 
since a nondeterministic algorithm can guess the states and 
verify transitions between consecutive states to find an accepting path. 
The ``automaton'' formulation made it easier to analyze 
logics on infinite words \cite{VW,VW08}.

\begin{theorem}
The satisfiability problem for \blintl\/ is in \Expspace.
\end{theorem}
\begin{proof}
The formula automaton has exponentially many states.
An accepting path may require going through an entire range of
global counter values, and with several such counters operating.
Hence an accepting path has to be guessed, written down and verified. 
This can be done in \Expspace. 
Corollary~\ref{coro:bthtl} showed that the sublogic \btlth\/ 
is already \Expspace-hard.
\end{proof}

\begin{corollary}
The model checking problem for \blintl\/ is \Pspace\/ in the size of 
the model and \Expspace\/ in the size of the formula.
\end{corollary}
\begin{proof}
Let $\alpha_0$ be a formula and $K$ a Kripke
structure. The above argument shows that for formula $\neg \alpha_0$ 
there is an exponential size formula automaton $M(\neg \alpha_0)$.
Verifying $K \models \alpha_0$ is equivalent to checking whether the
intersection of the languages corresponding to $K$ and $M(\neg \alpha_0)$
is nonempty. This can be done by a nondeterministic algorithm which uses
space logarithmic in the size of both the models. Since
$M(\neg \alpha_0)$ is exponentially larger than $\alpha_0$ we get the
upper bounds in the statement of the theorem, using Savitch's theorem.
The lower bounds are already known for Counting \ltl\/ \cite{LMP}.
\end{proof}

\subsection{Extensions}
\paragraph{Infinite words}
We note that our arguments are not affected
by whether we consider finite or infinite word models.
Hence our results carry over to the usual \ltl\/ setting of infinite words.

\paragraph{Finite group counting constraints}
We gave some details for the group counting constraints and it is easy to see 
that the results also hold when we add group counting constraints to \blintl. 
If we have a purely group counting logic without any threshold constraints, 
we can use the algebraic fact that a finite group 
has a generating set of logarithmic size to obtain a \Pspace\/ complexity 
with the syntax changed to refer to generators. These and other details 
are studied in a PhD thesis \cite{Sree}.

\paragraph{Finitely generated ordered group constraints}
Why did we not pursue our more ambitious logic with constraints over 
a finitely generated and discretely ordered group?

\begin{theorem}[Laroussinie, Meyer and Petonnet \cite{LMP}]
The satisfiability and model checking problems are undecidable for 
the logic with ordered group constraints.
\end{theorem}
\begin{proof}
The presence of the integer constants $c_i \in \Int$ allows easy programming
of the increment and decrement operations of a two-counter machine.
Hence the halting problem for these machines can be reduced to the 
satisfiability problem for the logic with ordered group constraints,
even over $\Int$.
\end{proof}

\paragraph{Branching time}
Our approach extends to \ctl\/ with counting constraints, studied by
Emerson, Mok, Sistla and Srinivasan \cite{EMSS} and 
Laroussinie, Meyer and Petonnet \cite{LMP-cctl}.
The \name{formula tree automaton} constructed uses states as above 
but a transition relation connects a state to several states,
the arity is determined by the number of existential $\until$/$\since$ 
requirements in a state \cite{JW95,VW08}.
We can prove that satisfiability is in \twoexptime.
\cite{LMP-cctl} obtained this upper bound by an exponential 
translation to ordinary \ctl\/ (with satisfiability in \Exptime\/
\cite{Emerson}) and a 
lower bound by describing an Alternating \Expspace\/ Turing machine.

\paragraph{Acknowledgements}
The authors would like to thank Simoni Shah and  A.V.~Sreejith,  
and acknowledge that a portion of the work surveyed here is drawn from 
their Ph.D. theses and associated papers.

\end{document}